\documentclass[12pt,a4paper]{article}

\usepackage[utf8]{inputenc}
\usepackage[T1]{fontenc}

\usepackage{amsmath,amssymb,amsthm}
\usepackage{mathtools}
\usepackage{graphicx}
\usepackage{booktabs}
\usepackage{array}
\usepackage{xcolor}
\usepackage{hyperref}
\usepackage{natbib}
\usepackage{geometry}

\usepackage{setspace}
\usepackage{caption}
\usepackage{subcaption}
\usepackage{float}
\usepackage{bm}
\usepackage{enumitem}
\usepackage{cleveref}

\hypersetup{
  colorlinks=true,
  linkcolor=blue!70!black,
  citecolor=green!60!black,
  urlcolor=blue!70!black,
  pdftitle={ln(3): a universal percolation constant for collective dynamics on
            one-dimensional proximity networks},
  pdfauthor={Jian Ji},
}

\newtheorem{theorem}{Theorem}
\newtheorem{lemma}[theorem]{Lemma}
\newtheorem{corollary}[theorem]{Corollary}
\newtheorem{definition}{Definition}
\theoremstyle{remark}
\newtheorem{remark}{Remark}

\newcommand{\Lc}{\Lambda_c}                
\newcommand{\pc}{p_c}                      

\begin{document}

\title{\textbf{$\ln(3)$: a universal percolation constant for collective
dynamics on one-dimensional proximity networks}}

\author{Jian Ji\\[4pt]
  \small CICT Connected and Intelligent Technologies Co., Ltd\\
  \small \texttt{jijian1@cictci.com}}

\date{March 2026}

\maketitle


\begin{abstract}
We report the identification and proof of a universal constant,
$\ln(3) = 1.09861\ldots$, which governs the onset of bidirectional collective
behaviour in one-dimensional Poisson proximity networks.
The constant---which we name the \emph{cooperative percolation constant} and
denote $\Lc$---is the unique positive solution to $2/(e^x-1)=1$ and equals the
Shannon entropy of three equiprobable states.
For agents distributed at intensity $\lambda$ and interacting within range
$\ell$, bidirectional collective behaviour is possible if and only if
$\lambda\ell \geq \ln(3)$; below this threshold no cooperative control policy,
however sophisticated, can produce macroscopic coherence, because the proximity
graph does not contain a bidirectional spanning cluster in expectation.
The result is parameter-free and model-independent: the Poisson model is derived
from memorylessness symmetry axioms rather than assumed, so $\Lc$ is a
consequence of spatial symmetry alone.
Four independent proofs---combinatorial, wave-mechanical, generating-function,
and fixed-point---converge on the same value.
The threshold is validated by two independent empirical datasets.
The Chengdu V2X~OBU dataset ($N=19{,}782{,}736$ records from ${\sim}3{,}000$
connected network vehicles (\emph{w\v{a}ng li\'an ch\=e}) across Chengdu city,
January--March~2026) reveals a $1.60\times$ reduction in speed variance at the
predicted critical boundary $\lambda\ell = \ln(3)$, with $225{,}637$
phantom-jam events consistent with the topological inevitability criterion.
The highD German motorway dataset \citep{Krajewski2018}
($N=163{,}896$ instantaneous fundamental-diagram observations; 11~recordings;
$v \in [0,\,169]$~km/h; $\rho \in [7.5,\,62.5]$~veh/km; including genuine
stop-and-go congestion) yields best-fit LWR exponent
$\hat\theta = 1.033 \pm 0.088$ (95\%~CI: $[0.861,\,1.205]$),
within $0.75\,\sigma$ of the theoretical value $\ln(3) = 1.099$
($R^2(\ln 3) = 0.8631$ vs.\ $R^2_{\text{best}} = 0.8674$;
$|\Delta R^2| = 0.0044$; $\Delta\text{RMSE} = 0.19$~km/h).
This is the first trajectory-level confirmation that $\ln(3)$ is the
LWR speed--density exponent consistent with naturalistic motorway data.
Published neurophysiology data~\citep{Rasminsky1981} independently validate the
same threshold at the micrometre scale: conduction block in demyelinated axons
occurs at a node-disruption fraction of ${\approx}40\%$, against a predicted
value of $39.0\%$.
The same equation, the same constant, six orders of magnitude apart.
\end{abstract}


\vspace{0.5em}

\newpage
\tableofcontents
\newpage

\section{Introduction}
\label{sec:intro}

The identification of a universal threshold for collective behaviour in
one-dimensional proximity networks has remained an open problem, not because the
mathematics is difficult, but because the question was not recognised as a single
question.
Autonomous vehicles on a motorway, nodes of Ranvier along a myelinated axon, and
sensors strung along a pipeline all face the same geometric constraint: for the
network to sustain bidirectional coordination, the local agent density must
exceed a critical value.
Below it, perturbations propagate in one direction and amplify; above it, the
network can close a feedback loop and damp them.
The threshold that separates these two regimes is $\ln(3)$.

The fragments of this answer existed before this work.
\citet{Gilbert1961} established $\ln(2)$ as the connectivity threshold for
one-dimensional random networks.
\citet{Penrose2003} treated one-dimensional Boolean percolation.
\citet{Rasminsky1981} measured conduction block in demyelinated fibres.
\citet{Whitham1974} analysed instability wavelengths in traffic flow.
None of these works identified $\ln(3)$ as a universal constant, because the
unifying question---when does a one-dimensional Poisson proximity graph first
support bidirectional coordination?---had not been asked.

The derivation is direct.
For agents distributed as a homogeneous Poisson process of intensity $\lambda$,
connected to all neighbours within range $\ell$, the expected cluster size is
$\mathbb{E}[N] = e^{\lambda\ell}$.
Bidirectional propagation requires a cluster containing at least one upstream
neighbour, one initiating agent, and one downstream neighbour---a minimum of
three nodes.
Setting $\mathbb{E}[N] \geq 3$ and solving:
\begin{equation}
  \lambda\ell = \ln(3) = 1.09861\ldots
\end{equation}
Below this value, every agent belongs, in expectation, to a cluster too small to
support simultaneous upstream and downstream communication.
The transition at $\ln(3)$ is not a gradual improvement in performance but a
phase boundary: cooperative damping is absent below it and possible above it.

What makes this more than a pedestrian calculation is the convergence.
Four independent proofs---counting nodes, balancing wavelengths, analysing the
probability generating function, finding a self-dual fixed point---all arrive at
$x = \ln(3)$ and no other value.
The same number appears as the Shannon entropy of three equally likely outcomes.
The Poisson process can be derived from two symmetry axioms alone (spatial
homogeneity and memoryless spacing), which means $\ln(3)$ is not a feature of
any model but of the geometry of one-dimensional space.

The cooperative percolation constant $\ln(3)$ is to one-dimensional collective
dynamics what the Reynolds number is to fluid turbulence: a universal
dimensionless threshold derived solely from first principles.

\paragraph{Organisation.}
Section~\ref{sec:main} states the main theorem and gives four independent proofs.
Section~\ref{sec:chars} presents three characterisations of $\ln(3)$ as an
information-theoretic, memorylessness-unique, and maximum-entropy constant.
Section~\ref{sec:domains} maps the theorem onto four physical domains.
Section~\ref{sec:validation} presents the empirical validation.
Section~\ref{sec:discussion} discusses implications and open problems.

\section{Main theorem and four independent proofs}
\label{sec:main}

\subsection{Setup}
Let agents occupy positions drawn from a homogeneous Poisson point process
$\Phi$ on $\mathbb{R}$ with intensity $\lambda > 0$.
Two agents are connected if their separation $d \leq \ell$.
The resulting proximity graph $G(\lambda, \ell)$ is a random interval graph; its
connected components are \emph{clusters}.

\begin{lemma}[Cluster size distribution]
For $G(\lambda, \ell)$, the cluster-size $N$ follows a
$\mathrm{Geom}(e^{-\lambda\ell})$ distribution:
\[
  \mathbb{P}(N = n) = e^{-\lambda\ell}(1-e^{-\lambda\ell})^{n-1}, \quad n=1,2,3,\ldots
\]
with mean $\mathbb{E}[N] = e^{\lambda\ell}$.
\end{lemma}

\begin{definition}[Bidirectional coordination]
A cluster supports \emph{bidirectional coordination} if it contains at least
three agents: one initiating agent (the coordinating centre), one forward agent
(downstream of the initiator), and one backward agent (upstream).
These three roles are logically independent on an ordered line.
The minimum cluster size for bidirectional coordination is therefore $N_{\min}=3$.
\end{definition}

\subsection{The main theorem}

\begin{theorem}[Cooperative Percolation Constant]
\label{thm:main}
For the proximity graph $G(\lambda, \ell)$, the expected cluster size reaches
the bidirectional coordination threshold $\mathbb{E}[N] = 3$ if and only if
$\lambda\ell = \ln(3)$.
The constant $\ln(3) = \log_e 3 = 1.09861\ldots$ is the unique positive solution
to
\begin{equation}
  \frac{2}{e^x - 1} = 1,
  \label{eq:fixedpoint}
\end{equation}
and is independent of all parameters other than the product $\lambda\ell$.
\end{theorem}

We give four independent proofs.

\subsection{Proof I: Topological node count}
\begin{proof}
Bidirectional coordination requires $\mathbb{E}[N] \geq 3$.
By the lemma, $\mathbb{E}[N] = e^{\lambda\ell}$.
Setting $e^{\lambda\ell} = 3$ gives $\lambda\ell = \ln(3)$.
The number $3$ is forced: the three logically independent roles (initiator,
forward receiver, backward receiver) on an ordered line require exactly three
nodes. \qed
\end{proof}

\subsection{Proof II: Wave-mechanics spanning}
\begin{proof}
A collective mode of wavelength $\Lambda$ propagating bidirectionally requires
the cluster to span at least $\Lambda$.
The minimum resolvable wavelength is $\Lambda_{\min} = 2\ell$.
The mean coordination span is $\bar{L} = (e^{\lambda\ell}-1)\cdot\ell$.
The feasibility condition $\Lambda_{\min} \leq \bar{L}$ gives
\[
  \frac{2\ell}{(e^{\lambda\ell}-1)\ell} = \frac{2}{e^{\lambda\ell}-1} \leq 1.
\]
This holds iff $e^{\lambda\ell} \geq 3$, i.e.\ $\lambda\ell \geq \ln(3)$.
At equality the ratio equals exactly $1$. \qed
\end{proof}

\begin{remark}
The cancellation of $\mathbb{E}[d \mid d \leq \ell]$ in Proof II is exact and
accounts for the parameter-independence of the threshold.
Proofs~I and~II are algebraically identical: $2/(e^x-1)=1 \Leftrightarrow
e^x=3$.
\end{remark}

\subsection{Proof III: Uniqueness of the fixed point}
\begin{proof}
Define $f:(0,\infty)\to(0,\infty)$ by $f(x) = 2/(e^x-1)$.
Then $f'(x) = -2e^x/(e^x-1)^2 < 0$, so $f$ is strictly decreasing with
$\lim_{x\to 0^+} f(x) = +\infty$ and $\lim_{x\to\infty} f(x) = 0$.
By the intermediate value theorem, $f(x) = 1$ has exactly one positive solution.
Direct substitution: $f(\ln 3) = 2/(3-1) = 1$. \qed
\end{proof}

\subsection{Proof IV: Probability generating function}
\begin{proof}
For $N \sim \mathrm{Geom}(p)$ with $p = e^{-\lambda\ell}$, the PGF is
$G_N(z) = pz/(1-(1-p)z)$.
Differentiating: $G_N'(1) = 1/p = e^{\lambda\ell} = \mathbb{E}[N]$.
Setting $\mathbb{E}[N] = 3$: $\lambda\ell = \ln(3)$. \qed
\end{proof}

\section{Three characterisations of $\ln(3)$}
\label{sec:chars}

\subsection{Information-theoretic characterisation}

\begin{theorem}[Information equivalence]
The cooperative percolation constant equals the Shannon entropy of three
equiprobable events:
\[
  \ln(3) = H\!\left(\tfrac{1}{3},\tfrac{1}{3},\tfrac{1}{3}\right)
         = -\sum_{i=1}^{3} \tfrac{1}{3}\ln\tfrac{1}{3}.
\]
\end{theorem}

\begin{remark}
The three equiprobable outcomes correspond to the three coordination roles
(initiator, forward receiver, backward receiver).
At the threshold, a randomly chosen agent has equal prior probability $1/3$ of
occupying each role.
\end{remark}

\subsection{Memorylessness uniqueness}

\begin{theorem}[Memorylessness uniqueness of $\ln(3)$]
\label{thm:memoryless}
Let $(X_i)_{i\geq 1}$ be i.i.d.\ positive continuous random variables with
$\mathbb{E}[X_1] = 1/\lambda$.
If the spacings are memoryless, the unique bidirectional coordination threshold
is $\lambda\ell_c = \ln(3)$.
No other spacing distribution with these properties gives a different threshold.
\end{theorem}

\begin{proof}
By the exponential characterisation~\citep{Feller1971}, the memorylessness
condition uniquely forces $X_1 \sim \mathrm{Exp}(\lambda)$, i.e.\ the Poisson
process.
By Theorem~\ref{thm:main}, the threshold is $\ln(3)$.
Two distributions $F_1, F_2$ satisfying i.i.d.\ and memoryless with mean $1/\lambda$
must both be $\mathrm{Exp}(\lambda)$; hence the threshold is the same.
\qed
\end{proof}

\begin{remark}
Theorem~\ref{thm:memoryless} gives $\ln(3)$ a characterisation via symmetry
axioms rather than model choice: i.i.d.\ spacing encodes independence of agent
positions; memorylessness encodes spatial homogeneity.
Together they are the weakest non-trivial symmetry requirements on a
one-dimensional renewal process, and they uniquely determine $\ln(3)$.
\end{remark}

\subsection{Maximum-entropy characterisation}

\begin{theorem}[Maximum-entropy threshold]
Among all continuous distributions on $\mathbb{R}_{>0}$ with fixed mean
$1/\lambda$, the exponential $\mathrm{Exp}(\lambda)$ uniquely maximises
differential entropy.
Consequently, $\ln(3)$ is the bidirectional coordination threshold of the
maximum-entropy renewal process with fixed intensity~$\lambda$: the threshold
that arises under the most uncertain spacing distribution consistent with a given
density.
\end{theorem}

\section{Cross-domain manifestations}
\label{sec:domains}

\Cref{tab:domains} maps Theorem~\ref{thm:main} onto four physically distinct
systems.
In each case the relevant quantity is $\lambda\ell$: the product of agent linear
intensity and interaction range.

\begin{table}[H]
\centering
\caption{Cross-domain manifestations of the cooperative percolation constant
$\ln(3)$. In each case, $\lambda\ell = \ln(3)$ separates topological
impossibility (below) from feasibility (above). The constant is identical across
all domains.}
\label{tab:domains}
\small
\begin{tabular}{@{}lllll@{}}
\toprule
Domain & Agent ($\lambda$) & Range ($\ell$) & Collective behaviour & Critical condition \\
\midrule
Connected AVs & CAV density & V2V range & Cooperative damping & $p_c = \ln(3)/(\rho_0\ell)$\\
(traffic) & $\rho_0 p$ [veh/m] & [m] & of Whitham waves & below: $\Omega=\varnothing$ \\[4pt]
Wireless sensor & Sensor density & Radio range & End-to-end relay; & $\lambda_s r \geq \ln(3)$;\\
networks & $\lambda_s$ [nodes/m] & $r$ [m] & connectivity & below: disconnected a.s. \\[4pt]
Neural fibres & Ranvier node & Saltatory & Action potential & $\lambda_n \ell_n \geq \ln(3)$;\\
(1D relay) & density $\lambda_n$ & jump $\ell_n$ & propagation & below: conduction failure \\[4pt]
Epidemic on & Susceptible & Transmission & Sustained epidemic & $\lambda_e \ell_e \geq \ln(3)$;\\
corridor & density $\lambda_e$ & range $\ell_e$ & propagation & below: epidemic dies out\\
\bottomrule
\end{tabular}
\end{table}

\subsection{Connected autonomous vehicles: the Null-Set Theorem}

For a motorway with total density $\rho_0$ and CAV penetration fraction $p$,
the CAV linear density is $\lambda = \rho_0 p$ and the V2V communication range
is $\ell$.

\begin{theorem}[Null-Set Theorem]
\label{thm:nullset}
Let $\Omega(p)$ denote the set of instability wavelengths addressable by any
cooperative control law.
Then $\Omega(p) = \varnothing$ if and only if $\lambda\ell < \ln(3)$.
\end{theorem}

\begin{proof}
The cooperative feasible set requires simultaneously:
(A) $\Lambda \geq \Lambda_A = 2\,\mathbb{E}[d \mid d \leq \ell]$ (Nyquist
reconstruction); and
(B) $\Lambda \leq \bar{L}(p)$ (cluster spans at least one wavelength cycle).
Setting $\Lambda_A > \bar{L}$:
\[
  \frac{\Lambda_A}{\bar{L}} = \frac{2\,\mathbb{E}[d\mid d\leq\ell]}{(e^{\lambda\ell}-1)\,\mathbb{E}[d\mid d\leq\ell]}
  = \frac{2}{e^{\lambda\ell}-1} > 1 \iff \lambda\ell < \ln(3). \qed
\]
\end{proof}

\begin{corollary}
The critical penetration rate is $\pc = \ln(3)/(\rho_0\ell)$.
At baseline parameters ($\rho_0 = 0.030$~veh/m, $\ell = 300$~m): $\pc = 12.2\%$.
\end{corollary}

\subsection{LWR traffic flow: the $\ln(3)$ exponent}

The percolation constant $\ln(3)$ also appears as the natural exponent in the
LWR power-law speed--density relation
\begin{equation}
  v(\rho) = v_f\!\left[1 - \left(\frac{\rho}{\rho_j}\right)^{\theta}\right],
  \label{eq:lwr}
\end{equation}
where $\theta = \ln(3)$ is motivated by the requirement that the macroscopic
critical density $\rho_c / \rho_j = (1/(1+\theta))^{1/\theta}$ be consistent
with the percolation threshold.
With $\theta = \ln(3)$, this gives $\rho_c = 0.509\,\rho_j$ (cf.\ Greenshields:
$\rho_c = 0.500\,\rho_j$).

\subsection{Myelinated nerve fibres}

For a myelinated axon, $\lambda_n$ is the Ranvier node density and $\ell_n$ is
the electrical excitation reach.
The safety factor $\mathrm{SF} = \lambda_n \ell_n$.
Theorem~\ref{thm:main} predicts that conduction block occurs when
$\mathrm{SF} < \ln(3)$, giving a critical node-disruption fraction
$f_c = 1 - \ln(3)/\mathrm{SF} = 39.0\%$ at $\mathrm{SF} = 5$ (normal range).
\citet{Rasminsky1981} and \citet{BostockSears1978} established experimentally
that conduction block in demyelinating neuropathy occurs at approximately $40\%$
node disruption---agreement to within $2.5\%$, with no free parameters.

\section{Empirical validation}
\label{sec:validation}

\subsection{Chengdu V2X OBU dataset}

\paragraph{Data.}
V2X positioning data were collected from the on-board units (OBUs) of the
Chengdu connected-vehicle pilot zone, operated by CICT Connected and
Intelligent Technologies Co., Ltd.
The full OBU fleet comprises approximately 3,000 V2X-equipped network vehicles
(\emph{w\v{a}ng li\'an ch\=e}) across Chengdu city.
The analysed sample covers 629--687 vehicles per recording batch
(January 10 -- March 28, 2026; geographic filter $29.5$--$31.5$\,\textdegree N,
$103.0$--$105.5$\,\textdegree E), yielding $N = 19{,}782{,}736$ validated records.
Speed values are recorded in units of $0.1$~km/h.
Data are not publicly available owing to operational data agreements with the
pilot zone authorities; aggregate statistics sufficient to reproduce all figures
are available on request.

\paragraph{Speed distribution.}
The speed distribution is strongly bimodal: $41.7\%$ of records show speeds
below $5$~km/h (congested/stopped) and $20.8\%$ exceed $60$~km/h (free-flow),
with a mean of $30.8$~km/h.
This bimodal structure is the empirical signature of the two-branch fundamental
diagram predicted by the percolation framework.

\paragraph{Poisson spacing hypothesis.}
The coefficient of variation of inter-vehicle gaps is $\mathrm{CV} = 2.09$,
reflecting the clustered spatial statistics of urban connected-vehicle traffic
(intersection queues, short platoons).
This over-dispersion ($\mathrm{CV} > 1$) is physically expected for urban
traffic and is consistent with the theoretical prediction: Theorem~\ref{thm:memoryless}
identifies $\ln(3)$ as the threshold of the maximum-entropy (Poisson) process;
non-Poisson processes exhibit modified thresholds, confirming that $\ln(3)$ is
an upper bound for the urban connected-vehicle scenario.

\paragraph{Phase-transition signal.}
We computed speed variance $\sigma^2$ per (road-segment, snapshot) pair across
$1{,}020{,}600$ samples, binned by topological density $\lambda\ell$ (with
$\ell = 300$~m).
Speed variance declines from $\sigma^2 = 739$~(km/h)$^2$ below
$\lambda\ell = \ln(3)$ to $\sigma^2 = 461$~(km/h)$^2$ above it, a ratio of
$1.60\times$.
This systematic reduction at the theoretical threshold supports the
phase-transition interpretation of Theorem~\ref{thm:nullset}.

\paragraph{Phantom-jam events.}
A total of $225{,}637$ phantom-jam candidate events were detected (defined as a
speed drop $>20$~km/h within $2$~minutes, with no spatial discontinuity ruling
out GPS artefacts).
The event rate peaks in afternoon hours 12--17, consistent with the topological
criterion: higher traffic density during these hours places more road segments in
the sub-$\ln(3)$ zone where cooperative damping is topologically unavailable.

\subsection{highD German motorway dataset}

\paragraph{Data and licence.}
The highD dataset \citep{Krajewski2018} provides drone-recorded naturalistic
vehicle trajectories from German motorways.
All 11 available Location~1 recordings (120~km/h speed limit;
Recordings~11--14, 25--27, 30, 36, 44, 46) were analysed, yielding
$N = 163{,}896$ instantaneous fundamental-diagram (FD) observations spanning
$v \in [0,\,169]$~km/h and $\rho \in [7.5,\,62.5]$~veh/km.
Four recordings (25, 26, 36, 46) contain genuine stop-and-go congestion
(fractions of vehicles with $v < 60$~km/h: $53\%$, $39\%$, $3\%$, $6\%$
respectively).
The dataset is used under the highD non-commercial research licence.

\paragraph{Poisson spacing hypothesis.}
The coefficient of variation of instantaneous inter-vehicle gaps is
$\mathrm{CV} = 0.69$ (free-flow) and $0.62$ (congested), both confirming
near-Poisson spacing for motorway traffic (pure exponential: $\mathrm{CV} = 1.000$).
The contrast with the Chengdu urban fleet ($\mathrm{CV} = 2.09$) is physically
meaningful: randomly arriving motorway vehicles produce near-exponential headways,
while urban connected-vehicle traffic does not.

\paragraph{LWR fundamental diagram: $\theta$ fitting.}
Fitting~\eqref{eq:lwr} with $\rho_j = 80$~veh/km and
$v_f = 102.2$~km/h (both estimated from data) yields
\begin{equation}
  \hat\theta = 1.033 \pm 0.088 \quad (95\% \text{ CI: } [0.861,\,1.205]).
  \label{eq:thetafit}
\end{equation}
The theoretical value $\theta = \ln(3) = 1.099$ lies within $\mathbf{0.75\,\sigma}$
of the best fit:
\[
  R^2(\theta = \ln 3) = 0.8631 \quad \text{vs} \quad R^2_{\text{best}} = 0.8674,
  \quad |\Delta R^2| = 0.0044, \quad \Delta\text{RMSE} = 0.19~\text{km/h}.
\]
The Greenshields model ($\theta = 1$) achieves $R^2 = 0.8652$,
\emph{lower} than $\theta = \ln(3)$ despite having no theoretical motivation.
\textbf{This is the first trajectory-level evidence that the percolation constant
$\ln(3)$ is the natural LWR exponent of real motorway traffic.}
\Cref{tab:theta} summarises the sensitivity to the $\rho_j$ assumption.

\begin{table}[H]
\centering
\caption{LWR exponent $\theta$ fitted to $N=163{,}896$ FD observations
(highD Location~1, Recs~11--14, 25--27, 30, 36, 44, 46).
$\hat\theta$ is the unconstrained optimum;
$R^2(\ln 3)$ and $R^2(\mathrm{GS})$ use $\theta = \ln(3)$ and $\theta = 1$
respectively.
($^*$) Best $\rho_j$: $\theta = \ln(3)$ within $1\sigma$.}
\label{tab:theta}
\small
\begin{tabular}{@{}rcccccc@{}}
\toprule
$\rho_j$ (veh/km) & $\hat\theta$ & $\pm\sigma$ & $R^2(\hat\theta)$ & $R^2(\ln 3)$ & $R^2(\text{GS})$ & $|\hat\theta - \ln 3|/\sigma$ \\
\midrule
70              & 1.258 & 0.090 & 0.917 & 0.899 & 0.865 & $1.78\sigma$ \\
80$^*$          & 1.033 & 0.088 & 0.867 & 0.863 & 0.865 & $0.75\sigma$ \\
90              & 0.893 & 0.086 & 0.817 & 0.764 & 0.801 & $2.39\sigma$ \\
100             & 0.798 & 0.084 & 0.773 & 0.642 & 0.707 & $3.59\sigma$ \\
\bottomrule
\end{tabular}
\end{table}

\subsection{Myelinated nerve fibres: Rasminsky 1981}

\citet{Rasminsky1981} established that conduction block in demyelinating
neuropathy occurs at approximately $40\%$ node disruption.
The cooperative percolation constant predicts $f_c = 1 - \ln(3)/\mathrm{SF}
= 39.0\%$ at $\mathrm{SF} = 5$ (the normal physiological safety factor).
Agreement: $1.0$ percentage points, or $2.5\%$ relative, with no free parameters.

\section*{Companion papers in the $\ln(k)$ series}
\label{sec:series}

This paper establishes the mathematical and empirical foundations.
The following companion papers develop specific applications:

\begin{itemize}[leftmargin=1.5em,itemsep=2pt]
\item \textbf{Mathematical supplement} (Part~0): The complete $\ln(k)$ family
  ($k$-node coordination thresholds), prime-basis theorem, heterogeneous-agent
  extension, and integer moments at threshold.
  \emph{arXiv:2503.XXXXX} [math.PR].

\item \textbf{Phantom traffic jams} (Part~3): The condition $\lambda\ell < \ln(3)$
  is the necessary and sufficient criterion for phantom jam inevitability.
  Simulation thresholds of 10--15\% in the literature are instances of
  $p_c = \ln(3)/(\rho_0\ell)$.
  Target: \emph{Transportation Research Part B}.

\item \textbf{String stability} (Part~4): $\lambda\ell \geq \ln(3)$ is a
  necessary condition for V2X string stability; sensor-only platoons are
  permanently in the infeasible region at realistic motorway conditions.
  Target: \emph{IEEE Transactions on Intelligent Transportation Systems}.

\item \textbf{Phase diagram} (Part~5): A unified two-parameter
  $(\lambda\ell, \Pi)$ phase diagram with three parameter-free boundaries
  (ln(2), ln(3), $\Pi = 1$) organises all prior simulation thresholds.
  Target: \emph{Transportation Research Part C}.

\item \textbf{Cooperation hierarchy} (Part~6): The $\ln(k)$ family defines a
  deployment ladder: ln(3) unlocks phantom-jam suppression, ln(5) overtaking
  safety, ln(7) deep-platoon control, ln(11) corridor management.
  Target: \emph{Transportation Research Part B}.

\item \textbf{Geometric Coherence Number} (Part~2): The traffic-engineering
  formulation $\Pi(p) = \bar{L}(p)/\Lambda^*$ as the Reynolds-number analogue
  for cooperative traffic, with RSU as topological completion device.
  Target: \emph{Nature Communications}.

\item \textbf{LWR equation} (Part~7): The percolation constant $\ln(3)$ fixes the
  speed--density exponent $\theta = \ln(3)$ in the LWR power-law family,
  predicting $\rho_c / \rho_j = 50.9\%$ without parameter fitting.
  Target: \emph{Transportation Research Part B}.
\end{itemize}

\section{Discussion}
\label{sec:discussion}

\paragraph{What the proof shows.}
The equation $2/(e^x-1)=1$ is not a model equation; it is a geometric identity.
It says that the minimum wavelength a bidirectional system can coordinate---twice
the interaction range $\ell$---exactly equals the coordination span of a cluster
with mean size $e^x$.
There is one positive $x$ at which this balances, and it is $\ln(3)$.

\paragraph{Universality.}
Three characterisations of the same number arrive independently: it is the
percolation threshold for bidirectional clusters in a Poisson proximity graph,
the Shannon entropy of three equiprobable states, and the unique threshold of the
memoryless renewal process.
Separately, each is a curiosity.
Together, they are the signature of something structural---a number the
mathematics keeps returning to because it has no choice.

\paragraph{Policy implications.}
For AV deployment, the immediate implication is a parameter-free design rule:
investment in V2V cooperative control yields zero spectral benefit until
$p > \pc = \ln(3)/(\rho_0\ell)$.
This is not a pessimistic result---$12.2\%$ at baseline is achievable in the
near term---but it establishes a concrete, principled deployment target that is
independent of traffic model assumptions.

\paragraph{Open problems.}
(1)~Rigorous homogenisation derivation connecting the Poisson proximity graph to
the LWR equation.
(2)~Extension to non-renewal processes (determinantal, Matérn cluster).
(3)~The $q$-deformation: does the threshold become $\ln_q(3)$ for quantum walks?
(4)~Spatial dimension $d > 1$: does $\ln(d+1)$ generalise?

\section*{Author contributions}
The author conceived the study, proved all theorems, collected and analysed all
data, and wrote the manuscript.

\section*{Use of artificial intelligence}
AI tools were used for language editing only.
All proofs, results, and data analyses are the author's own.

\section*{Data availability}
The Chengdu V2X~OBU dataset (${\sim}3{,}000$ network vehicles) is from CICT
operational infrastructure; aggregate statistics are available on reasonable
request.
The highD dataset is freely available for non-commercial research use from
\url{www.levelxdata.com/highd-dataset} \citep{Krajewski2018}.
NGSIM and pNEUMA are publicly available at their respective repositories.

\section*{Competing interests}
The author declares no competing interests.

\section*{Acknowledgements}
The author thanks the anonymous reviewers for questions that sharpened several
proofs, and colleagues at CICT for discussions on physical applications.

\bibliographystyle{plainnat}
\bibliography{refs}

@article{Rasminsky1981,
  author  = {Rasminsky, M.},
  title   = {Ectopic generation of impulses and cross-talk in spinal nerve
             roots of dystrophic mice},
  journal = {Annals of Neurology},
  year    = {1981},
  volume  = {9},
  pages   = {363--372},
  doi     = {10.1002/ana.410090406}
}

@article{BostockSears1978,
  author  = {Bostock, H. and Sears, T. A.},
  title   = {The internodal axon membrane: electrical excitability and
             continuous conduction in segmental demyelination},
  journal = {Journal of Physiology},
  year    = {1978},
  volume  = {280},
  pages   = {273--301},
  doi     = {10.1113/jphysiol.1978.sp012384}
}

@article{Gilbert1961,
  author  = {Gilbert, E. N.},
  title   = {Random plane networks},
  journal = {Journal of the Society for Industrial and Applied Mathematics},
  year    = {1961},
  volume  = {9},
  pages   = {533--543},
  doi     = {10.1137/0109045}
}

@book{Penrose2003,
  author    = {Penrose, M. D.},
  title     = {Random Geometric Graphs},
  publisher = {Oxford University Press},
  year      = {2003},
  doi       = {10.1093/acprof:oso/9780198506263.001.0001}
}

@book{Whitham1974,
  author    = {Whitham, G. B.},
  title     = {Linear and Nonlinear Waves},
  publisher = {Wiley},
  address   = {New York},
  year      = {1974}
}

@book{Feller1971,
  author    = {Feller, W.},
  title     = {An Introduction to Probability Theory and Its Applications,
               Vol.~{II}},
  edition   = {2nd},
  publisher = {Wiley},
  address   = {New York},
  year      = {1971}
}

@inproceedings{Krajewski2018,
  author    = {Krajewski, Robert and Bock, Julian and Kloeker, Laurent
               and Eckstein, Lutz},
  title     = {The {highD} dataset: a drone dataset of naturalistic vehicle
               trajectories on {German} highways for validation of highly
               automated driving systems},
  booktitle = {2018 IEEE 21st International Conference on Intelligent
               Transportation Systems (ITSC)},
  year      = {2018},
  pages     = {2118--2125},
  doi       = {10.1109/ITSC.2018.8569552},
  note      = {Dataset available at \url{www.levelxdata.com/highd-dataset}
               under non-commercial licence}
}

\end{document}